%% file: main.tex
\documentclass[submission,copyright,creativecommons]{eptcs}
\providecommand{\event}{ICE 2026} 

\usepackage{iftex}

\ifpdf
  \usepackage{underscore}         
  \usepackage[T1]{fontenc}        
\else
  \usepackage{breakurl}           
\fi

\usepackage{pdfx}
\usepackage{listings}
\usepackage{graphicx}
\usepackage{euscript}
\usepackage{tabularx}
\usepackage{xurl}

\usepackage{subcaption}
\usepackage{multirow}
\usepackage{amsthm}
\usepackage{amsfonts}
\usepackage{amssymb}
\usepackage{amsmath}
\usepackage{booktabs}
\usepackage{blkarray, bigstrut}

\input{macros.tex}

\newtheorem{definition}{Definition}[section]

\newtheorem{property}{Property}[section]

\newtheorem{proposition}{Proposition}[section]

\newtheorem{claim}{Claim}[section]

\begin{document}

\title{A Practical Approach To Verifying Structural Invariants\\In Colored Petri Nets}


\author{Lorenzo Capra
\institute{
{Department of Informatics}\\
    {Universit{\`a} degli Studi di Milano}, Italy}}



\def\titlerunning{Practical Verification Of Structural Invariants In CPN}
\def\authorrunning{L. Capra}


\maketitle

\begin{abstract}
Structural analysis is a core method in Petri Net (PN) research, offering a perspective complementary to state-space techniques while avoiding many of their limitations. It is well studied for classical PNs but far less for High-Level Petri Nets (HLPN). Symmetric Nets (SN), a common HLPN formalism, use compact annotations to encode behavioral symmetries, enabling the construction of a symbolic reachability graph (and a lumped Markov chain in stochastic SN) and the execution of symbolic discrete-event simulations.

During the past two decades, structural techniques tailored to SN have been developed, notably supported by the \SNex\ tool. This tool implements a formal calculus designed for the computation of symbolic structural relations, including, but not limited to, conflict relations and causal dependencies. Here, we focus on using this calculus to verify semi-automatically symbolic structural invariants, a task currently feasible only for certain restricted SN subclasses. We focus specifically on (semi)flows and briefly discuss an approach through which a flow generative family can be generated, at least theoretically. We further briefly outline a framework for the formal verification of a broader class of invariant properties. An extended formulation of the SN formalism is employed, which demonstrably satisfies the closure property with respect to fundamental functional operators. The core concepts are elucidated by means of representative examples throughout the exposition.

\end{abstract}



\section{Introduction And Related Work}
\label{intro}
\input{introduction}
\input{SNdefinition}


\section{Symbolic (Semi)Flows: Basic Definitions And Properties}
\label{sec:semiflow-def}
As stated above, hereinafter we implicitly consider constant-size functions. 
Let \vect{H} be the $|P |\times |T|$ incidence matrix of an SN: $\vect{H}[p,t] = \Out{p,t} - \Inp{p,t}$. The symbol '$\cdot$' denotes the usual product of function matrices based on the composition of functions. Assume that places and transitions are indexed; we can use the notation $\vect{H}[i,j]$ instead of $\vect{H}[p_i,t_j]$

\begin{definition}[P-(semi)flow] Let $\vect{I}_P$ be a non-null $1 \times |P|$ vector of functions such that $\forall p$, $\vect{I}_P[p] \in \fset{\cdom(p)}{\Gbag{\cdom'}}$, and $\vect{I}_P[p] \not\equiv 0 \Rightarrow \cdom' \leq \cdom(p)$.\\
$\vect{I}_P$ is a P-flow (P-invariant) if and only if $\vect{I}_P \cdot \vect{H} \equiv \vect{0}$.\\
A $P$-flow such that $\forall p$, and $\vect{I}_P[p] \in \fset{\cdom(p)}{\Bag{\cdom'}}$ is called $P$-semiflow.
\label{def:P-semif}
\end{definition}

A (semi)flow with minimal support (that is, which cannot be obtained as a linear combination of others) is said to be minimal. We expressly refer to minimal flows in the following.

Although general flows may reveal interesting properties, in our examples, we mostly focus on semiflows, which have an intuitive interpretation.
If $\vect{m}_0$ is the initial SN marking and $\vect{I}_P$ is a semiflow, then for each reachable marking $\vect{m}$, $\vect{I}_P \cdot \vect{m} = \vect{I}_P \cdot \vect{m}_0$. 
From a computational point of view, there is no difference.

Analogously, T-flows are "right" solutions of a homogeneous system that has $\vect{H}$ as a symbolic coefficient matrix. Let us focus on semiflows.

\begin{definition}[T-semiflow] Let $\vect{I}_T$ be a non-null $|T| \times 1$ vector of functions such that $\vect{I}_T[t] \in \fset{\cdom'}{\Bag{\cdom(t)}}$ and $\vect{I}_T[t] \not\equiv 0 \Rightarrow \cdom' \leq \cdom(t)$.\\
$\vect{I}_T$ is a T-semiflow if and only if $\vect{H} \cdot \vect{I}_T   \equiv \vect{0}$.
\label{def:T-semif}
\end{definition}

The T-semiflows also admit an intuitive interpretation. Let $\vect{I}_T$ be a T-semiflow, and let $\vect{A}$ belong to the image of $\vect{I}_T$ ($\vect{A}[t] \in \cdom(t)$); then, from each marking $\vect{m}$, any sequence of transition instances conforming to $\vect{A}$ leads back to $\vect{m}$.

These symbolic flows represent invariant properties of both quantitative and qualitative nature with respect to the entire ESN. By the way, it is also possible to limit the flow calculation to a specific subnet.
Their definitions differ in subtle ways from earlier formulations (e.g. \cite{Eva2007}), which construct a common minimal super-domain by inserting auxiliary “dummy" functions $All$ into arc functions, thus producing a model that is equivalent with respect to bisimulation criteria. 
Using the largest sub-domains (our choice) makes it easier to interpret symbolic flows and relate them to the conventional flows of the unfolding.

\subsection{Properties}
Under the implicit assumption that all functions we consider have a constant size, an intuitive but pertinent property establishes a connection between symbolic and conventional flows. 
\begin{definition}
    \label{def:skel}
The skeleton of an ESN $\EuScript{N}= (P, T, I, O, \ldots)$ is a PT net $(P, T, F)$.
$F : P \times T \cup T \times P \rightarrow \Nat$ is such that $F(p,t) = size(\Inp{p,t})$, $F(t,p) = size(\Out{p,t})$.  
\end{definition}

\begin{property}
\label{propr:skel}
Let $\EuScript{N}_{skel}$ be the skeleton of the ESN $\EuScript{N}= (P, T, \ldots)$.  If the $1 \times |P|$ vector $\vect{I}_P$ is a symbolic P-flow for $\EuScript{N}$ (Def. \ref{def:P-semif}) then
the derived vector in $\Int$ $\vect{I}'_P$ such that $\vect{I}'_P[p] = size(\vect{I}_P[p])$ is a conventional P-flow for $\EuScript{N}_{skel}$.
\end{property}

An analogous connection exists between the symbolic T-flows of $\EuScript{N}$ and those of $\EuScript{N}_{skel}$. 
This property enables a substantial reduction of the search space for candidate flows. A potential application of this result to the determination of a generative basis of semiflows will be briefly addressed in Section \ref{sec:basis}. 

Let \vect{H} be the incidence matrix of $\EuScript{N}$. Its transpose $\vect{H}^t$ is a $|T| \times |P|$ matrix such that $\vect{H}^t[i, j] = \vect{H}[j,i]^t$ (this operation is different from the conventional matrix transposition). Analogously, $\vect{I}_P^t$ is a vector $|P| \times |1|$ such that $\vect{I}_P^t[j] = \vect{I}_P[j]^t$.
The following property states that there exist two equivalent methodologies to validate candidate flows. This can be convenient from a computational perspective. An analogous property is valid for T-flows.

\begin{property}
\label{propr:eqT-P-flows}
Let $\vect{I}_P$ be a non-null $1 \times |P|$ vector. Then
$$\vect{I}_P \cdot \vect{H} \equiv \vect{0}
\Leftrightarrow  \vect{H}^t \cdot \vect{I}_P^t \equiv \vect{0}$$
\end{property}

This property also sets a relationship between the P and T-flows. A P-flow of $\EuScript{N}$ is a T-flow for the dual net $\EuScript{N}^d$, obtained by swapping places with transitions (reverting to the arc directions) and replacing the original arc functions with their transposes. And vise versa, a T-flow of $\EuScript{N}$ is a P-flow of $\EuScript{N}^d$.

\begin{figure}[!h]
\begin{center}
\includegraphics[width=0.8\textwidth]{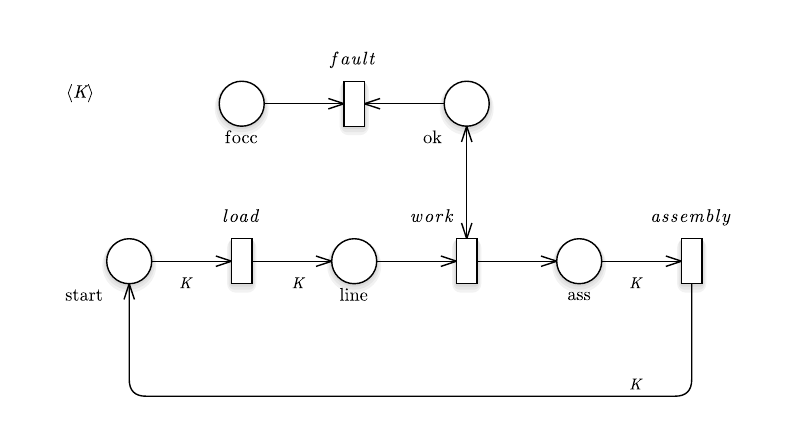}
\caption{Skeleton of the PL model (Fig. \ref{fig:PL}) }
\label{fig:skeleton}
\end{center}
\end{figure}

\section{Semi-Flows Verification: A Few Examples}
\label{sec:semiflow-ver}

In this section, we validate candidate P and T-semiflows using several examples. Each of these examples outlines specific aspects of the ESN syntax. 
Two are drawn from the literature, while the others have been constructed for illustrative purposes. The calculations were carried out semi-automatically with the support of \SNex. In particular, the current implementation still does not natively support the composition in $\lang$. Nevertheless, one can simulate this composition via that in $\overline{\lang}$, provided the arc functions are (re)expressed as pairwise disjoint linear combinations $\sum_i \lambda_i F_i$ where terms $F_i$ are guarded tuples that map to multisets in which all multiplicities are equal to one (which is the natural choice of modelers). 
The single calculation steps take negligible time.
All models were edited using the GreatSPN package \cite{GreatSPN}, which also supports the unfolding of the SN and the computation of conventional (numerical) semiflows. A translator from the \GreatSPN\ format (\url{https://github.com/greatspn/SOURCES}) to the \SNex\ format can be downloaded from the homepage.

\paragraph{Example 1}
Fig. \ref{fig:PLflows} presents the  incidence matrix of the SN representing a distributed production line (Fig. \ref{fig:PL}) together with potential (semi)flows, according to property \ref{propr:skel} (with the T-vector also displayed as a row hereinafter); The SN skeleton is shown in Fig. \ref{fig:skeleton}.

These candidate symbolic flows can be efficiently validated using the \SNex\ command-line interface (CLI). The color domains associated with the transitions, as well as the domains of the functions, are explicitly represented. It should be noted that the same symbol (e.g., $\tuple{l}$) can denote functions defined in different domains.

\begin{figure}[!ht]
\[
\begin{blockarray}{ccccc}
 & load_{PL} & work_{PL,L} & assembly_{PL} & fault_{PL,L} \\
\begin{block}{c[cccc]}
\p{start} & -\tuple{All_L} &  & \tuple{All_L} & \bigstrut[t] \\
\p{line} & \tuple{pl, All_L} & -\tuple{pl, l} & &  \\
\p{ass} &  & \tuple{pl, l} & -\tuple{pl, All_L} & \\
\p{ok} & & & & -\tuple{pl, l}\\
\p{focc} & & & & -\tuple{pl}\bigstrut[b]\\
\end{block}
\end{blockarray}\vspace*{-1.25\baselineskip}
\]

\[
\begin{blockarray}{cccccc}
& \p{start} & \p{line} & \p{ass} & \p{ok} & \p{focc} \\
\begin{block}{c[ccccc]}
\vect{I}^1_P  & \tuple{l}_{L} & \tuple{l}_{PL,L} & \tuple{l}_{PL,L} &  &  \\
\end{block}
\begin{block}{c[ccccc]}
\vect{I}^{2}_P & &  & & -\tuple{pl}_{PL,L} & \tuple{pl}_{PL} \\
\end{block}
\end{blockarray}\vspace*{-1.3\baselineskip}
\]

\vspace*{0.35\baselineskip}
\[
\begin{blockarray}{ccccc}
& load_{PL} & work_{PL,L} & assembly_{PL} & fault_{PL,L} \\
\begin{block}{c[cccc]}
\vect{I}_T  & \tuple{pl}_{PL} & \tuple{pl, All_L}_{PL} & \tuple{pl}_{PL} &    \\
\end{block}
\end{blockarray}\vspace*{-1.3\baselineskip}
\]

\caption{Incidence Matrix and P- and T-semiflows of SN in Fig. \ref{fig:PL}}
\label{fig:PLflows}
\end{figure}

Semiflow $\vect{I}^1_P$ admits a straightforward interpretation. It indicates that the distributed PLs (the lower part of the figure) are conservative with respect to the color class $L$: for example, if the initial marking contains a multiple of $K$ distinct tokens (representing workpieces) in place $\p{start}$, this quantity is preserved throughout execution. 
$\vect{I}^2_P$ is the only instance of a flow that does not match a semiflow reported in the paper. This implies that the $PL$ component associated with the place $\p{ok}$ decreases consistently with the marking of the place $\p{focc}$.
Because the SN places are covered by (semi)flows, the SN is structurally bounded (and color-safe, provided that the initial marking is color-safe). 

Even more interestingly, the T-semiflow identifies a production \emph{cycle} within a given PL, represented by the symbol $pl$: this cycle consists of an occurrence of $load$, $K$ occurrences of $work$, and an occurrence of $assembly$.
Observe that these symbolic semiflows are intrinsically parametric, since no multiplicities appear either in $\vect{H}$ or in the semiflow expressions.

We now briefly analyze the relationship between symbolic and conventional (semi)flows, with reference to the unfolding of Fig. \ref{fig:PL} reported in the Appendix (see Fig. \ref{fig:unfolding}). 
Consider first $\vect{I}^1_P$. We note that the components associated with the places $\p{line}$ and $\p{ass}$ correspond to the projection onto class $L$ of the color domain $PL \times L$. Consequently, for a given color of $L$ bound to $l$, there exist as many conventional semiflows as there are possible combinations of $PL$-colors in these two places. This yields a total of $N^2 \cdot K$ conventional semiflows that instantiate $\vect{I}^1_P$. 
For example, when $N = 2$ and $K = 2$, we obtain the eight P-semiflows shown in the Appendix (Fig. \ref{fig:unfolding}).

The symbolic T-semiflow, in turn, corresponds to \(N\) conventional semiflows, each associated with a distinct binding of \(pl\) (see Fig.~\ref{fig:unfolding}).
\begin{figure}[!t]
    \centering
    \includegraphics[width=\linewidth]{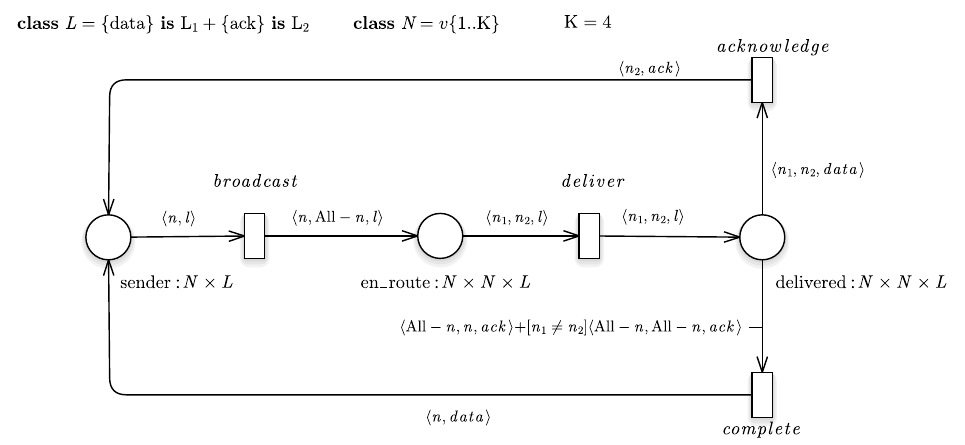}
    \caption{An ESN representing a broadcast protocol ($ack \equiv All_{L2}$, $data \equiv All_{L1}$)}
    \label{fig:broadcast}

\includegraphics[width=0.9\textwidth]{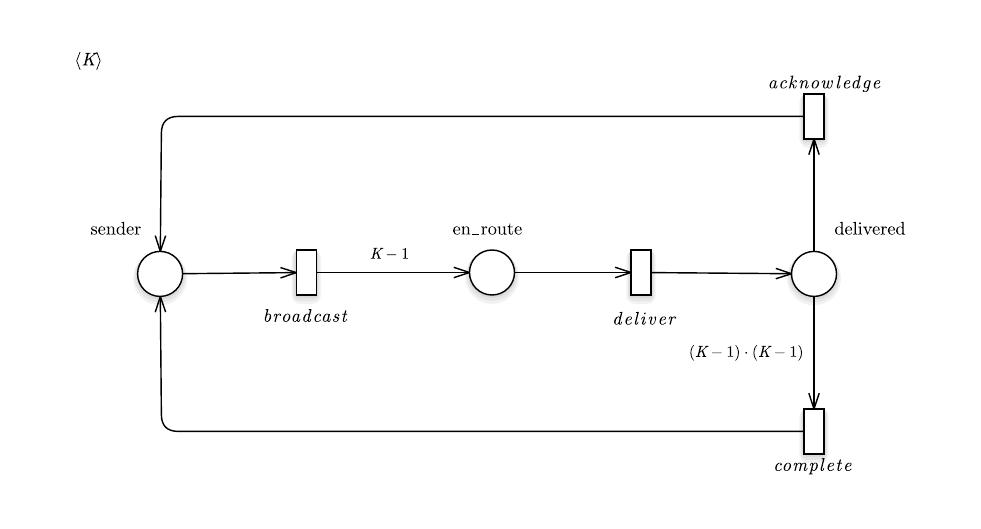}
\caption{Skeleton of the broadcast model (Fig. \ref{fig:broadcast})}
\label{fig:skeletonbroad}
\end{figure}

\begin{figure}[!h]
\centering
\[
\begin{blockarray}{ccccc}
 & broadcast_{N,L} & deliver_{N,N,L} & acknowledge_{N,N} & complete_{N} \\
\begin{block}{c[cccc]}
\p{sender} & -\tuple{n,l} &  & \tuple{n_2,ack} & \tuple{n,data} \bigstrut[t] \\
\p{en\_route} & \tuple{n, All - n, l} & -\tuple{n_1, n_2, l} & &  \\
\p{delivered} &  & \tuple{n_1, n_2, l} & -\tuple{n_1, n_2, data} & -\tuple{All - n, n, ack} \\
 &  &  &  & - [n_1 \neq n_2]\tuple{All - n,All -  n, ack}\\
\end{block}
\end{blockarray}\vspace*{-0.75\baselineskip}
\]

\vspace*{0.25\baselineskip}
\[
\begin{blockarray}{ccccc}
& broadcast & deliver & acknowledge & complete \\
\begin{block}{c[cccc]}
  & (\tuple{n,data}+ &  (\tuple{n,All - n, data} + & \tuple{n,All - n}_{N} &  \tuple{n}_{N}  \\
\vect{I}_T & \tuple{All - n,ack})_{N} & \tuple{All - n, n, ack} + &  &    \\
 &  & [n_1 \neq n_2]\tuple{All - n,All -  n, ack})_{N} &  &    \\
\end{block}
\end{blockarray}\vspace*{-0.3\baselineskip}
\]

\caption{Incidence Matrix and T-semiflow of SN in Fig. \ref{fig:broadcast}}
\label{fig:T-flowbroad}
\end{figure}

\paragraph{Example 2}
The second example in this section is a variant of one presented in \cite{Camilli2021}, and it models a broadcast communication protocol. This model necessitates the utilization of the extended SN syntax. The ESN shown in Fig. \ref{fig:broadcast} (and the following) uses more sophisticated functions. The class $N$ denotes the nodes in a network, while the class $L$ is divided into two singleton subclasses that encode \emph{data} and \emph{acknowledge} messages, respectively. The protocol operates as follows: when a message $l$ is to be transmitted from a source node $n$ (the place $\p{sender}$), it is sent to all other nodes in the network (transition $broadcast$). Once a \emph{data} message has been delivered to a node, that node schedules an acknowledgment using the same broadcast mechanism. The protocol terminates successfully when all nodes, including the original sender, have received an acknowledgment.

A notable component of the SN is the function $\Inp{\p{delivered}, complete}$, previously described, which represents a sophisticated synchronization condition: all other nodes have given an acknowledgment to node $n$ (the original sender), and each of these nodes has received an acknowledgment from all nodes except $n$ and itself. 
 Within this function, a filter is applied to extract all pairs of distinct nodes from the parametric Cartesian product \((All - n) \times (All - n)\).
 
The ESN skeleton is shown in Fig. \ref{fig:skeletonbroad}. Since it does not allow any P-semiflows, the ESN consequently has no symbolic P-semiflows either.
In contrast, it admits the following conventional T-semiflow (with $K = |N|$): 
\[
[broadcast: K,\; deliver: K \cdot (K - 1),\; acknowledge: K - 1,\; complete: 1].
\]
The candidate symbolic semiflow shown in Fig. \ref{fig:T-flowbroad} is in agreement with this. To confirm that it is indeed a T-semiflow, we must resolve nontrivial compositions like the one below, which arises when multiplying the second row of \vect{H} by $\vect{I}_T$. The symbolic calculus reduces all expressions $e$ to a normal form $\hat{e} \in \EuScript{L}$.

The equivalence of terms can be established purely syntactically since the calculus guaranties that, whenever $e \equiv 0$, it follows that $\hat{e} = 0$.

\vspace{-10pt}
\begin{align*}
\tuple{n, All - n, l}  \circ \tuple{All - n,ack} &\equiv \\
\tuple{n, All, l}  \circ \tuple{All - n,ack} - \tuple{n, n, l}  \circ \tuple{All - n,ack} &\equiv \\
\tuple{All - n, All, ack}  - [n_1 == n_2]\, \tuple{All - n, All - n, ack} &\equiv \\
[n_1 \neq n_2] \, \tuple{All - n, All - n, ack}  + \tuple{All - n, n, ack} & \\
\end{align*}

We can directly verify that the row-by-column products are zero and therefore $\vect{I}_T$ is a semiflow. It corresponds to the completion of a basic protocol cycle: for any $n$ representing the original sender, and ignoring the specific order, it includes one $broadcast$ of a $data$ message from $n$, and $broadcast$s of an $ack$ message from all the other nodes; a set of $deliver$ events of the $data$ message from $n$ to every other node, a set of $deliver$ events of an $ack$ from each other node to $n$, and a set of $deliver$ events of an $ack$ from any node other than $n$ to any other node distinct from both $n$ and itself; the emission of an $aknowledge$ by every node other than $n$; and finally a $complete$ event, which represents the synchronization upon reception of all expected $ack$ messages by $n$ and by all other nodes.

\paragraph{Example 3}
The SN shown in Fig. \ref{fig:exe2} uses an ordered color class $C$ and a color domain in which this class is replicated. The possible symbolic semiflows (Fig. \ref{fig:exe2flows}) correspond to the semiflows of the underlying skeleton, namely $[\p{p}_1:1 \,\, \p{p}_2: |C| \,\, \p{p}_3:1]$ and $[t_1:1 \,\, t_2: 1 \,\, t_3:|C|-1]$. Using the algorithm defined in \cite{pnse2021} verify, for example, that the product of $\vect{I}^1_P$ by the first column of $\vect{H}$ is null. We obtain:
\begin{eqnarray*}
   \tuple{c_2} \circ \tuple{!c, c} + N \tuple{c} \circ -\tuple{c} + \tuple{c_1} \circ \tuple{c, All -!c} &\equiv& \tuple{c} - N \tuple{c} + (N-1) \tuple{c} 
\end{eqnarray*}
\noindent We can similarly prove that the other row-by-column products are nullified.

Regarding the interpretation, $\vect{I}^1_P$ expresses that for any reachable marking, the colors appearing in the second component of the 2-tuples in $\p{p}_1$, together with those in the first component of the tuples in $\p{p}_3$, always sum up to $N = |C|$ times the multiset of colors in $\p{p}_2$. Dually, $\vect{I}^2_P$ can be read as follows: in all reachable markings $\vect{m}$, the colors in the first component of the 2 tuples in $\p{p}_1$ plus those in the second component of the tuples in $\p{p}_3$ always form the entire color set $C$, where each color occurs with multiplicity $|\vect{m}(\p{p}_2)|$. 
The T-semiflow, instead, asserts that every firing sequence consisting of one occurrence $\tuple{c}$ of $t_1$ and $t_2$, together with all occurrences $\tuple{c, c'}$ of $t_3$ for which $c' \neq \, !c$, returns to the starting marking.

\begin{figure}
    \centering
    \includegraphics[width=0.7\linewidth]{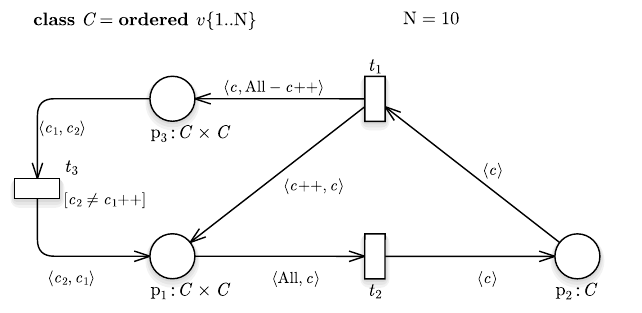}
    \caption{An SN with an ordered class (++ is !)}
    \label{fig:exe2}

\[
\begin{blockarray}{cccc}
 & t_{1_{C}} & t_{2_{C}} & t_{3_{C,C}}  \\
\begin{block}{c[ccc]}
\p{p}_1 & \tuple{!c, c} & -\tuple{All, c} &  \tuple{c_2, c_1}[c_2 \neq \, !c_1] \bigstrut[t] \\
\p{p}_2 & -\tuple{c} & \tuple{c} &   \\
\p{p}_3 & \tuple{c, All - \, !c} &  & -\tuple{c_1, c_2}[c_2 \neq \, !c_1]  \\
\end{block}
\end{blockarray}\vspace*{-1.25\baselineskip}
\]

\[
\begin{blockarray}{cccc}
& \p{p}_1 & \p{p}_2 & \p{p}_3  \\
\begin{block}{c[ccc]}
\vect{I}^1_P  & \tuple{c_2}_{C,C} \quad & \mathrm{N} \tuple{c}_{C} \quad & \tuple{c_1}_{C,C}   \\
\end{block}
\begin{block}{c[ccc]}
\vect{I}^2_P  & \tuple{c_1}_{C,C} &  \tuple{All} & \tuple{c_2}_{C,C} \\
\end{block}
\end{blockarray}\vspace*{-1.3\baselineskip}
\]

\vspace*{0.35\baselineskip}
\[
\begin{blockarray}{cccc}
& t_1 & t_{2} & t_{3}  \\
\begin{block}{c[ccc]} \vect{I}_T  & \tuple{c}_{C} & \tuple{c}_{C} & \tuple{c, All - \, !c}_{C}  \\
\end{block}
\end{blockarray}\vspace*{-1.3\baselineskip}
\]

\caption{Incidence Matrix and semiflows of SN in Fig. \ref{fig:exe2}}
\label{fig:exe2flows}
\end{figure}

\begin{figure}
    \centering
    \includegraphics[width=0.7\linewidth]{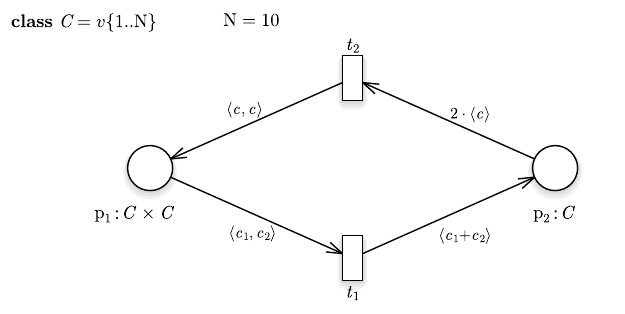}
    \caption{An SN with functions with non-one coefficients}
    \label{fig:exe3}
\end{figure}

\begin{figure}[!ht]
\[
\begin{blockarray}{ccc}
 & t_{1_{C,C}} & t_{2_{C}}   \\
\begin{block}{c[cc]}
\p{p}_1 & -\tuple{c_1,c_2}_{C,C} & \tuple{c,c}_{C}  \bigstrut[t] \\
\p{p}_2 & \tuple{c_1 + c_2}_{C,C} & -2\tuple{c}_{C}    \\
\end{block}
\end{blockarray}\vspace*{-1.25\baselineskip}
\]

\[
\begin{blockarray}{ccc}
& \p{p}_1 & \p{p}_2   \\
\begin{block}{c[cc]}
\vect{I}_P  & \tuple{c_1 + c_2}_{C,C} &  \tuple{c}_{C} \\
\end{block}
\end{blockarray}
\vspace*{-1.3\baselineskip}
\]

\vspace*{0.2\baselineskip}
\[
\begin{blockarray}{ccc}
& t_1 & t_{2}   \\
\begin{block}{c[cc]}
\vect{I}_T  & \tuple{c,c}_{C} & \tuple{c}_{C}  \\
\end{block}
\end{blockarray}\vspace*{-1.3\baselineskip}
\]

\caption{Incidence Matrix and semiflows of SN in Fig. \ref{fig:exe3}}
\label{fig:exe3flows}
\end{figure}

\paragraph{Example 4}
In the small SN shown in Fig. \ref{fig:exe3}, we first encounter arc functions whose images are multisets with multiplicities greater than one. This raises no theoretical issues for our calculus; instead, it is forbidden in \cite{Eva2007}. All possible semiflows can be checked directly. For example, $\vect{I}^1_P$ states that for $\p{p}_1$, the sum of the colors in the two components of its 2-tuples always matches the multiset of colors in $\p{p}_2$.

\section{Generating A Semi-Flow Basis}
\label{sec:basis}
The capability to formally verify symbolic semiflows is particularly valuable, as modelers typically possess a priori knowledge regarding the expected system behavior and seek to corroborate this knowledge through automated verification support. Nevertheless, the capability to construct a semiflow generative family represents an even more critical issue from both theoretical and practical standpoints. As mentioned previously, existing theoretical approaches have achieved only partial success, as they rely on imposing rather stringent syntactic constraints on the admissible classes of arc functions.  

In this section, we concisely present a hybrid empirical–theoretical methodology that, instead of relying on advanced algebraic frameworks, exploits Property \ref{propr:skel} under the standing assumption that all arc functions have a constant size.

\begin{claim}
The set of arc functions in $\lang_{\cdom_{\ColClasses},\cdom_{\ColClasses}'}$ of a certain size $k$ is finite and can be computed explicitly.
\end{claim}

The preceding claim is naturally justified by the fact that the collection of class-functions with domain $\cdom_{\ColClasses}$ is uniquely fixed, and the symbolic color-tuples in the domain $\cdom'_{\ColClasses}$ of cardinality $k$ can be generated by an algorithmic enumeration. In general, this enumeration problem, which is theoretically nontrivial, may become computationally expensive when the parameter $k$ is large and-or the color domains are composite. From this point of view, Property \ref{propr:eqT-P-flows} may help. 

We argue that the difficulty of this task can be significantly reduced by leveraging certain properties of SN arc functions. In particular, any arc function $F$ can be expressed in the form $\sum_i \lambda_i T_i$, with $\lambda_i \in \Nat$, where each $T_i$ is a tuple (optionally preceded by a filter) that produces multisets in which all element multiplicities are equal to one.

In this paper, we do not introduce any algorithm; instead, we demonstrate how to use the previous claim to construct a semiflow generative family, drawing on a simple example from Section \ref{sec:semiflow-ver}.
More challenging cases arise when considering the other examples in that section.

Consider the SN depicted in Fig.~\ref{fig:PL}, whose underlying skeleton is provided in the Appendix (Fig.~\ref{fig:skeleton}). The P-semiflow basis of this skeleton consists of the two vectors $[\p{start}:1 \quad \p{line}:1 \quad \p{ass}:1]$ and $[\p{ok}:1 \quad \p{focc}:1]$. In order to derive the corresponding symbolic P-semiflows, we must first select a common subdomain for the places belonging to the support of each semiflow. In this case, there is one candidate subdomain for each of the two P-semiflows, namely $L$ and $PL$, respectively. These classes are neither ordered nor partitioned; consequently, the corresponding admissible class-functions are $\{l,All_L\}$ and $\{pl,All_{PL}\}$. Depending on the place under consideration, these functions may have different domains (for instance, $l_L$ for the place $\p{start}$ and $l_{PL,L}$ for the place $work$); we let the context resolve this ambiguity. Hence, by concentrating on the first P-semiflow and observing that the only arc functions of cardinality one that can be constructed are $\tuple{l}$ if $K > 2$ and either $\tuple{l}$ or $\tuple{All - l}$ if $K = 2$, we directly conclude that the symbolic P-semiflow $\vect{I}_p^1$ depicted in Fig.~\ref{fig:PLflows} constitutes part of the basis for $K > 2$. If $K = 2$, the P-vector $[\p{start}:\tuple{All - l}_{L} \quad \p{line}:\tuple{All - l}_{PL,L} \quad \p{ass}:\tuple{All - l}_{L}]$ also constitutes a semiflow, which is equivalent to the previously defined one. We do not elaborate further on this aspect here. The basis of P-semiflows is completed by $\vect{I}_p^2$ (Fig.~\ref{fig:PLflows}), for which analogous comments can be made.

The unique minimal T-semiflow of the SN skeleton is given by  
\[
[load:1 \;\; work:K \;\; assembly:1],
\]  
And therefore, the only admissible subdomain of transition domains is \(PL\).  
This subdomain will serve as the common domain for all functions that make up the symbolic T-semiflow.
With respect to the transition \(work\), we must determine arc functions of cardinality  $K := |L|$. The only feasible candidates for the corresponding semiflow component are $\tuple{pl, All_L}$ and, in the specific case where $N (:= |PL|) = K = 2$,
the additional candidate \(2 \langle All - pl, All_L \rangle\).

\section{Colored Flows}
\label{sec:cflow}

Symbolic flows capture notable properties, both qualitative (the types of colors) and quantitative (the number of colors). In general, however, they comply with rather restrictive, conservative patterns.

It is frequently advantageous to verify qualitative and quantitative invariants separately. In particular, the preservation of specific structural properties of colored place markings can be independent of the precise cardinality of tokens present, for example in the case of unbounded models. This is exactly the purpose of \emph{colored semiflows}, hereafter referred to simply as Csemiflows.

Let $\vect{H}^+$ and $\vect{H}^-$ be the $|P|\times|T|$ matrices defined entrywise by
\[
\vect{H}^+[p,t] = \overline{\Out{p,t} \ominus \Inp{p,t}}, 
\qquad
\vect{H}^-[p,t] = \overline{\Inp{p,t} \ominus \Out{p,t}}.
\]
Equivalently, for a given color instance of transition $t$, the entries $\vect{H}^+[p,t]$ and $\vect{H}^-[p,t]$ represent, respectively, the sets of color tuples (tokens) that are added and removed from place $p$.

\begin{definition}[P-Csemiflow] Let $\vect{I}_{P-c}$ be an $1 \times |P|$ vector of functions such that $\vect{I}_{P-c}[p] : \cdom(p) \rightarrow 2^{\cdom'}$ and $\vect{I}_{P-c}[p] \not\equiv 0 \Rightarrow \cdom' \leq \cdom(p)$.\\
$\vect{I}_{P-c}$ is a P-Csemiflow if and only if $\vect{I}_{P-c} \cdot \vect{H}^+  \,\equiv\, \vect{I}_{P-c} \cdot \vect{H}^-$.
\label{def:P-c-semif}
\end{definition}

\noindent T-Csemiflows are defined analogously. 

\paragraph{Example 5} Consider the SN depicted in Fig. \ref{fig:exe5}.  The underlying skeleton admits one P-semiflow: $[\p{p}_1:2 \quad \p{p}_2: 1 \quad \p{p}_3:3]$; It is straightforward to verify that there are no symbolic semiflows that match it (using property \ref{propr:skel} and considering the functions of proper size). Each transition is equipped with a guard, and there is also a self-loop between $t_1$ and $\p{p}_2$. A transition guard is naturally propagated over the incident arc functions. To construct the matrices $\vect{H}^+$ and $\vect{H}^-$, we must rewrite the difference $\vect{H}[\p{p}_2,t_1] := \Out{\p{p}_2,t_1} - \Inp{\p{p}_2,t_1}$ as a pairwise disjoint algebraic sum, which can always be achieved in our framework. In this simple situation, we obtain (omitting the tuple notation; $g_1 \equiv \Guards(t_1)$):
$$(-2 c_1 + c_2)[g_1] \equiv  c_2[c_1 \neq c_2][g_1] -(c_1[c_1 == c_2] +2 c_1[c_1 \neq c_2])[g_1] $$

When considering the support of positive and negative terms, we remove multiplicities and obtain $\vect{H}^+[\p{p}_2,t_1] = c_2[c_1 \neq c_2 \wedge g_1]$, $\vect{H}^-[\p{p}_2,t_1] = c_1[g_1]$. (Hereinafter, in this subsection, we omit the multiset support notation for simplicity.). The other elements of $\vect{H}^+$ and $\vect{H}^-$ are simply derived.

At this stage, we can verify (by applying the calculus in $\overline{\lang}$) whether the two vectors $\vect{I}_P^i$ depicted in Fig.~\ref{fig:exeCflow} actually define colored semiflows. The first would, in that case, correspond to the conservation of the colors in place $\p{p}_2$ and the first component of the 2-tuples in place $\p{p}_1$, respectively, with respect to the subclass $C_1$.  
However, this verification fails. 
Consider, for example, the product of $\vect{I}_P^1$ with the first column of $\vect{H}^+$ and, analogously, with the first column of $\vect{H}^-$:
\begin{eqnarray*}
\tuple{c}[c \in C_1] \circ \tuple{c_2}[g_1 \wedge c_1 \neq c_2]  & \equiv& \tuple{c_2}[g_1 \wedge c_1 \neq c_2] \\
\tuple{c_2}[c_2 \in C_1] \circ \tuple{c_1,c_2}[g_1] + \tuple{c}[c \in C_1] \circ \tuple{c_1}[g_1 ] &\equiv& \tuple{c_2}[g_1 ] + \tuple{c_1}[g_1 ]\\
\end{eqnarray*}
\noindent the resulting expressions, of course, are not equivalent.
In contrast, we easily verify that $\vect{I}_P^2$ is a Csemiflow: the set of colors in place $\p{p}_3$ and the first component of the 2-tuples in $\p{p}_1$ remain invariant.

An interesting concluding observation is that the calculus in $\overline{\lang}$ implemented in \SNex\ is inherently parametric with respect to the cardinalities of color classes, which can be specified by linear constraints; for example, $|C| \geq 3$. Consequently, we are able to validate Csemiflows not only in a purely symbolic manner but also in a parametric setting. For instance, the aforementioned results hold irrespective of the cardinality of the subclass \(C_1\).  
A detailed investigation of this parametric analysis is beyond the scope of the present work.

\begin{figure}
    \centering
    \includegraphics[width=0.75\linewidth]{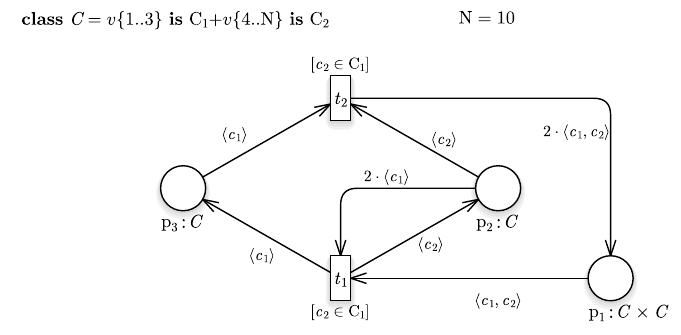}
    \caption{SN with a split color class and guards}
    \label{fig:exe5}
\end{figure}

\begin{figure}[!ht]
\centering
\begin{minipage}{0.48\linewidth}
\centering
\[
\begin{blockarray}{ccc}
\vect{H}^+ & t_{1_{\,C,C}} & t_{2_{\,C,C}}   \\
\begin{block}{c[cc]}
\p{p}_1 &  & \tuple{c_1,c_2}[g_2]  \bigstrut[t] \\
\p{p}_2 & \tuple{c_2}[g_1 \wedge c_1 \neq c_2] &     \\
\p{p}_3 &  \tuple{c_1}[g_1] &     \\
\end{block}
\end{blockarray}
\]
\end{minipage}\hfill
\begin{minipage}{0.48\linewidth}
\centering
\[
\begin{blockarray}{ccc}
\vect{H}^- & t_{\,1_{C,C}} & t_{\,2_{C,C}}   \\
\begin{block}{c[cc]}
\p{p}_1 & \tuple{c_1,c_2}[g_1] & \bigstrut[t] \\
\p{p}_2 & \tuple{c_1}[g_1] & \tuple{c_2}[g_2]    \\
\p{p}_2 &  & \tuple{c_1}[g_2]    \\
\end{block}
\end{blockarray}
\]
\end{minipage}

\[
\begin{blockarray}{cccc}
& \p{p}_1 & \p{p}_2  & \p{p}_3 \\
\begin{block}{c[ccc]}
\vect{I}_P^1  & \tuple{c_2}[c_2 \in C_1]_{C,C} \, \, \, &  \tuple{c}[c \in C_1]_{C} \, \, \, & \\
\end{block}
\begin{block}{c[ccc]}
\vect{I}_P^2  & \tuple{c_1}_{C,C} \, \, \, &  \, \, \, & \tuple{c}_{C}\\
\end{block}
\end{blockarray}
\]

\caption{$\vect{H}^+$ and $\vect{H}^-$ matrices and P-Csemiflow of SN in Fig. \ref{fig:exe5}}
\label{fig:exeCflow}
\end{figure}

\section{Verifying More General Invariants}
\label{sec:genInv}
Using a simple yet nontrivial example, we then briefly illustrate how the SN structural calculus implemented in \SNex\ can be employed to rigorously verify properties that go beyond symbolic or colored semiflows. With this method, we broaden the class of models that can be examined (including, for instance, unbounded models) and account for inhibitor arcs, which substantially enhance the expressive power of the SN modeling framework.

The approach is based on a well-established principle. Our goal is to establish the property $\EuScript{P}$, which is formally defined in an extended language $\lang^+$ of $\lang$ (or $\overline{\lang}^+$ if we disregard the quantitative aspects).

\begin{enumerate}
    \item We first show that $\EuScript{P}$ is satisfied in the initial marking (which can also be described symbolically).
    \item We then show that if $\EuScript{P}$ is met in a generic marking $\vect{m}$, then it is also met in any $\vect{m}'$ such that $\vect{m} [(t, c) > \vect{m}'$, for all $t \in T, c \in \cdom(t)$.
\end{enumerate}

Our focus is on the second point, which is the most significant. The essential requirement is that this proof be carried out at a symbolic level, that is, by operating on symbolic transition instances that compactly encode the corresponding concrete instances. To this end, and in analogy with SMT-based techniques, we employ the \SNex\ rewriting engine to infer new axioms (rewriting rules), which are then dynamically incorporated into the underlying theory. Currently, this inference process is performed manually. However, we plan to extend the \SNex\ CLI to make the overall workflow semi-automated.

From a technical point of view, we must enrich $\lang$ ($\overline{\lang}$) with symbols that denote parametric (multi)sets over specified color domains. We use symbols like $Z_{\cdom_\ColClasses}$ to denote a \emph{parametric} (multi-)set in $\cdom_\ColClasses$, where $\cdom_\ColClasses$ denotes any subdomain made up of color classes in $\ColClasses$.
This symbol can be interpreted as a function with a null (or neutral) domain; that is, $ \bullet \rightarrow Bag[\cdom_\ColClasses]$ or $ \bullet \rightarrow 2^{\cdom_\ColClasses}$, depending on whether we are considering $\EuScript{L}^+$ or $\overline{\EuScript{L}}^+$.

The functional operators of the calculus are then applied in a uniform manner. For example, we may form expressions such as $\tuple{Z_{\cdom_\ColClasses}^1, Z_{\cdom_\ColClasses}^2}$, $Z_{\cdom_\ColClasses}^1 \ \mu\ Z_{\cdom_\ColClasses}^2$, or $F \circ Z_{\cdom_\ColClasses}^1$, where $\mu \in \{+(-),\cap, \ominus \}$, $F \in \fset{\cdom_\ColClasses}{D}$.
    
\begin{figure}[t]
    \centering
    \includegraphics[width=0.7\linewidth]{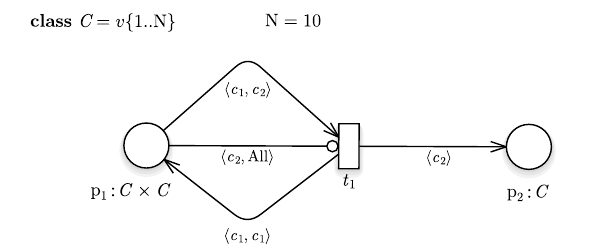}
    \caption{An SN transition with inhibitor edge}
    \label{fig:exe6}
\end{figure}

\paragraph{Example 6}
Let us clarify the concept using the SN in Fig. \ref{fig:exe6}. Our goal is to show that the first component of the 2-tuples of colors in place $\p{p}_1$ and the colors in place $\p{p}_2$ are disjoint, under the assumption that this holds in the initial marking $\vect{m}_0$. Demonstrating this directly by inspecting the color functions is not straightforward at all. We can represent an arbitrary SN marking $\vect{m}$ as: $[\vect{m}(\p{p}_1): Z_{C,C}^1 \quad \vect{m}(\p{p}_2): Z_{C}^1]$. The invariant property is (we suppose that '$\circ$' takes priority over'$\cap$' and the latter over '+'):
$${c_1} \circ Z_{C,C}^1 \cap Z_{C}^1 \equiv 0$$

The first natural step consists of incorporating the corresponding rewriting rule ${c_1} \circ Z_{C,C}^1 \cap Z_{C}^1 \longrightarrow 0$ into the framework. Subsequently, we syntactically encode within $\vect{m}$ the enabling condition associated with transition $t_1$. In principle, this procedure can be fully automated.
A generic symbolic instance of $t_1$ is given by $(t_1,\tuple{Z_{C}^2, Z_{C}^3})$, where $Z_{C}^2$ and $Z_{C}^3$ denote parametric colors in $C$ that are bound to variables (projections) $c_1$ and $c_2$, respectively.  

The constraints imposed by the input and inhibitor arc functions can be represented in the form of rewrite rules. With respect to the input constraint, we obtain:
\[
Z_{C,C}^1 \longrightarrow Z_{C,C}^2 + \tuple{Z_{C}^2, Z_{C}^3}, \quad |Z_{C}^2| = 1, |Z_{C}^3| = 1.
\]

 By substituting this encoding into the invariant expression, which is supposed to hold in $\vect{m}$, and then applying the composition rules together with standard multiset properties—such as the distributivity of composition over multiset sum and the following:

$$
(A + B) \cap C \equiv 0 \Rightarrow A \cap C \equiv 0 \wedge B \cap C \equiv 0,
$$
we finally obtain this expression, which can be simplified:
\begin{eqnarray*}
{c_1} \circ (Z_{C,C}^2 + \tuple{Z_{C}^2, Z_{C}^3}) \cap Z_{C}^1 & \equiv  {c_1} \circ Z_{C,C}^2 \cap Z_{C}^1 + Z_{C}^2 \cap Z_{C}^1  \equiv \emptyset \\
\end{eqnarray*}

\vspace{4pt}
\noindent From which we infer the rules: i) ${c_1} \circ Z_{C,C}^2 \cap Z_{C}^1\longrightarrow 0$, ii) $Z_{C}^2 \cap Z_{C}^1  \longrightarrow 0$.

The constraint imposed by the inhibitor arc function is formalized as follows:
\begin{eqnarray*}
\tuple{Z_{C}^3, All} \cap (Z_{C,C}^2 + \tuple{Z_{C}^2, Z_{C}^3}) < \tuple{All, All} & \cong \\
\tuple{Z_{C}^3, All} \cap (Z_{C,C}^2 + \tuple{Z_{C}^2, Z_{C}^3})  \equiv 0 & \cong \\
\tuple{Z_{C}^3, All} \cap Z_{C,C}^2 \equiv 0 \, \wedge \, \tuple{Z_{C}^3, All} \cap \tuple{Z_{C}^2, Z_{C}^3} \equiv 0 &
\end{eqnarray*}
\noindent From which we infer the rules: iii) $Z_{C}^3 \cap c_1 \circ Z_{C,C}^2 \longrightarrow 0$,  iv) $Z_{C}^3 \cap Z_{C}^2 \longrightarrow 0$.

\vspace{5pt}
The marking expression we obtain by symbolically firing the instance $(t_1,\tuple{Z_{C}^2, Z_{C}^3})$ in \vect{m} is: 

\vspace{5pt}
\hspace{2.5cm}$[\vect{m}'(\p{p}_1): Z_{C,C}^2 + \tuple{Z_{C}^2,Z_{C}^2} \quad \vect{m}'(\p{p}_2): Z_{C}^1 + Z_{C}^3]$.

\vspace{5pt}
\noindent Finally, it remains to establish the following result:
\begin{align*}
  {c_1} \circ (Z_{C,C}^2 + \tuple{Z_{C}^2,Z_{C}^2}) \cap (Z_{C}^1 + Z_{C}^3)  \equiv 
  ({c_1} \circ Z_{C,C}^2 + Z_{C}^2) \cap (Z_{C}^1 + Z_{C}^3)  \equiv 0 & \cong  \\
  {c_1} \circ (Z_{C,C}^2) \cap Z_{C}^1  \equiv 0 \wedge {c_1} \circ (Z_{C,C}^2) \cap Z_{C}^3  \equiv 0 \wedge
  Z_{C}^2 \cap  Z_{C}^1  \equiv 0 \wedge
  Z_{C}^2 \cap  Z_{C}^3  \equiv 0 &   
\end{align*}

\noindent which can be directly inferred from axioms i)-iv).

As a final remark, we note that this methodology can also be employed to validate candidate symbolic P-semiflows (or colored flows). However, in this case, validation is achieved through a more involved and less streamlined procedure than the approaches presented in Sections \ref{sec:semiflow-ver} and \ref{sec:cflow}.

\section{Conclusions} We have demonstrated that the symbolic structural calculus for Symmetric Nets (SN), developed over the past two decades to derive structural dependencies between SN nodes and implemented in the \SNex\ tool, can also be used to validate a broad class of symbolic (semi)flows for extended Symmetric Net (ESN) models. Furthermore, we have conducted a preliminary investigation into the derivation of a semiflow basis by exploiting a natural correspondence with the conventional semiflows of the ESN skeleton. Finally, we have outlined how the calculus can be used to formally verify more general classes of structural invariants.


Practically, we plan to extend \SNex\ by (1) fully supporting the composition of multiset functions, (2) adding an interactive inference mechanism that feeds newly derived axioms into the \SNex\ rewriting engine, and (3) enabling smooth interoperability with other tools, especially \GreatSPN.
Theoretically, we study whether generalized inverses of arc functions (extended linearly) can be expressed symbolically to construct a generative family of symbolic flows, potentially building on the preliminary results of this paper.


  

\bibliographystyle{eptcs}
\bibliography{biblio}

\newpage

\appendix

\section*{Appendix: proofs, unfolding, and conventional semiflows}

\proof{Property \ref{propr:skel}. Consider the $j^{th}$ column of $\vect{H}$. The row-by-column product $\sum_{p \in P} \vect{I}_P[p] \circ \vect{H}[p,t_j]$ results in the null function (Hp). Since the function linear extension is also constant-size, $\vect{I}_P[p] \circ \vect{H}[p,t_j]$ is constant-size and $size(\vect{I}_P[p] \circ \vect{H}[p,t_j]) = size(\vect{I}_P[p]) \cdot (size(\Out{p,t_j}) - size(\Inp{p,t_j})$.\\ Therefore, the algebraic sum corresponding to the conventional row-by-column product is zero; thus, $\vect{I}'_P$ is a conventional flow of $\EuScript{N}_{skel}$.

\noindent Property \ref{propr:eqT-P-flows}. It follows directly from the definition of the transpose of a matrix, together with the fundamental rule connecting the transpose of a function with composition:
$g \circ f \equiv f^t \circ {g}^t$    
}


   
\begin{figure}[h]
\begin{center}
\includegraphics[width=0.8\textwidth]{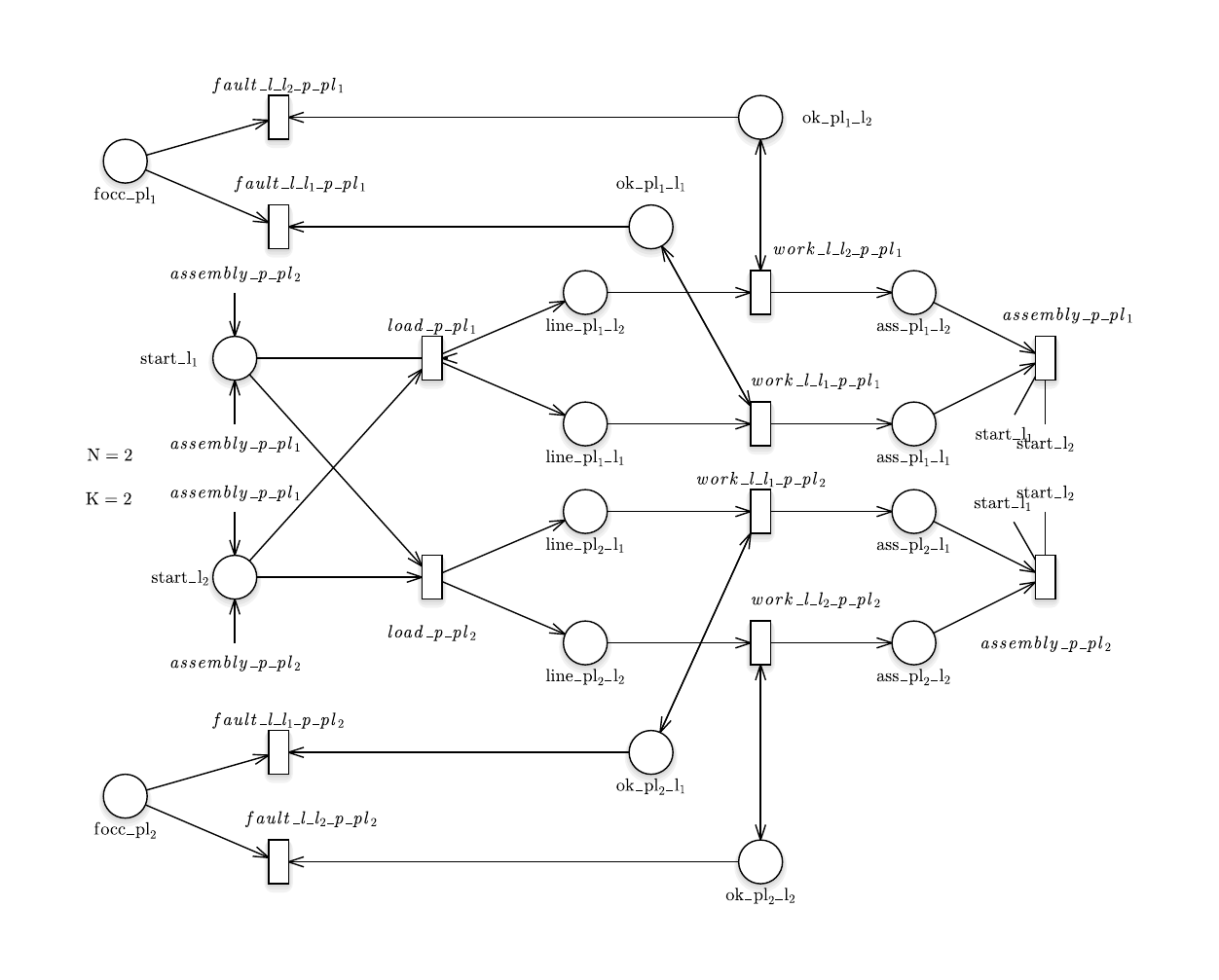}

\caption{Unfolding of the PL model ($N = 2$, $K = 2$)}
\label{fig:unfolding}
\end{center}
\end{figure}

\begin{center}
{\footnotesize
\begin{tabular}{lllll}
\hline
 \p{start\_l_1} & \p{line\_pl_1\_l_2} & \p{line\_pl_2\_l_2} & \p{ass\_pl_1\_l_2} & \p{ass\_pl_2\_l_2} \\
 \p{start\_l_1} & \p{line\_pl_1\_l_1} & \p{line\_pl_2\_l_2} & \p{ass\_pl_1\_l_1} & \p{ass\_pl_2\_l_2} \\
 \p{start\_l_1} & \p{line\_pl_1\_l_2} & \p{line\_pl_2\_l_1} & \p{ass\_pl_1\_l_2} & \p{ass\_pl_2\_l_1} \\
 \p{start\_l_1} & \p{line\_pl_1\_l_1} & \p{line\_pl_2\_l_1} & \p{ass\_pl_1\_l_1} & \p{ass\_pl_2\_l_1} \\
 \p{start\_l_2} & \p{line\_pl_1\_l_2} & \p{line\_pl_2\_l_2} & \p{ass\_pl_1\_l_2} & \p{ass\_pl_2\_l_2} \\
 \p{start\_l_2} & \p{line\_pl_1\_l_1} & \p{line\_pl_2\_l_2} & \p{ass\_pl_1\_l_1} & \p{ass\_pl_2\_l_2} \\
 \p{start\_l_2} & \p{line\_pl_1\_l_2} & \p{line\_pl_2\_l_1} & \p{ass\_pl_1\_l_2} & \p{ass\_pl_2\_l_1} \\
 \p{start\_l_2} & \p{line\_pl_1\_l_1} & \p{line\_pl_2\_l_1} & \p{ass\_pl_1\_l_1} & \p{ass\_pl_2\_l_1}\\
\end{tabular}
\vspace{1em}

\begin{tabular}{llll}
\hline
$\mathit{load\_p\_pl_2}$ & $\mathit{work\_l\_l_1\_p\_pl_2}$ & $\mathit{work\_l\_l_2\_p\_pl_2}$ & $\mathit{assembly\_p\_pl_2}$ \\
$\mathit{load\_p\_pl_1}$ & $\mathit{work\_l\_l_1\_p\_pl_1}$ & $\mathit{work\_l\_l_2\_p\_pl_1}$ & $\mathit{assembly\_p\_pl_1}$ \\
\hline
\end{tabular}

\captionof{table}{P- and T-semiflows of the Petri net shown in Fig. \ref{fig:unfolding}.}
}
\end{center}

\end{document}

%% file: macros.tex
\newcommand{\cdom}{{\ensuremath{\EuScript{D}}}}

\newcommand{\Inp}[1]{I(#1)}
\newcommand{\Out}[1]{O(#1)}
\newcommand{\Inh}[1]{H(#1)}

\newcommand{\ColClasses}{\ensuremath{\Sigma}}
\newcommand{\Guards}{\ensuremath{\Phi}}

\newcommand{\mk}[1]{\ensuremath{\mathbf{#1}}}

\newcommand{\vect}[1]{\ensuremath{\mathbf{#1}}}

\newcommand{\Bag}[1]{\ensuremath{Bag[{#1}]}}
\newcommand{\Gbag}[1]{\ensuremath{Bag^*[{#1}]}}
\newcommand{\fset}[2]{\ensuremath{\{#1 \rightarrow #2\}}}

\newcommand{\p}[1]{\ensuremath{\mathrm{#1}}}

\newcommand{\SNex}{\textsf{SNexpression}}
\newcommand{\GreatSPN}{\textsf{GreatSPN}}

\newcommand{\tuple}[1]{\langle{#1}\rangle}
\newcommand{\lang}{\EuScript{L}}
\newcommand{\Nat}{\mathbb{N}} 
\newcommand{\Int}{\mathbb{Z}}

%% file: introduction.tex
High-Level Petri Nets (HLPNs) \cite{HLPN} constitute a broad family of Petri Net-based formalisms (including, among others, Algebraic Petri Nets, Predicate/Transition Nets and Colored Petri Nets) that extend basic Petri Nets (PN) \cite{ReisigPN} by admitting structured tokens rich in data in place of indistinguishable ones. In HLPNs, both the places and transitions—the nodes of the underlying bipartite graph—are typed or parameterized, so that each node may admit multiple distinct instantiations. The arcs of an HLPN are labeled with expressions which, depending on the specific formalism, can be algebraic terms, functions, or other forms of symbolic annotation. By enabling such data and type structures, HLPNs provide a more abstract and compact representation of complex concurrent systems.

In Colored Petri Nets (CPNs), according to the original definition of Jensen \cite{Jensen1991,CPN09}, the nodes of the net are associated with \emph{finite color domains} . Each arc is annotated with a function from the color domain of the incident transition into multisets over the color domain of the incident place. A CPN can be converted into an equivalent Place/Transition (P/T) net called \emph{unfolding} \cite{Jensen1991,ReisigPN}. In realistic scenarios, this unfolding becomes excessively complex or even computationally intractable, and the subsequent back-interpretation of the analysis results may be rendered considerably more difficult.

Symmetric Nets (SN) \cite{CDFH93}, formerly known as Well-formed Nets, constitute a class of CPNs characterized by a dedicated syntactic structure explicitly tailored to exploit the behavioral symmetries of the underlying model during its analysis.
The domains associated with SN nodes are  Cartesian products of basic \emph{color classes}. A color class can be partitioned into static subclasses (or ordered circularly); colors in a subclass represent entities with similar behavior.  
SN functions encode the color domain structure of the nodes; arcs are annotated with expressions that are composed of tuples of base color functions.
Despite this syntactic specialization, SNs have been formally shown to be equally expressive as the original CPN formalism \cite{Jensen1991}, in the strict sense that any CPN can be systematically transformed into an equivalent SN.
Using a symbolic initial marking for an SN, together with a symbolic firing rule, one builds a compact Symbolic Reachability Graph, which is strongly bisimilar to the ordinary Reachability Graph. For Stochastic SN, this graph yields a lumped Markov chain; alternatively, it is possible to perform simulations using a symbolic discrete‑event simulation framework \cite{GreatSPN}.

A distinctive advantage of Petri Nets is that one can obtain meaningful properties directly from their structure, without constructing the state space; extending this capability efficiently to CPNs, without using unfolding (that is, symbolically), is far from straightforward. SN structural analysis has made substantial progress in the past two decades. The theoretical foundations are rooted in \cite{CAPRA2005}, which defines a language to symbolically express and calculate the principal structural relations (conflict, causal connection, mutual exclusion) \cite{QEST2015}. The elements of this language are similar to the SN arc functions, but show expanded expressivity by using guards also as function prefixes.
The {\SNex} software tool \cite{Capra2020} (\url{www.di.unito.it/~depierro/SNexpression}), featuring a command-line interface (CLI) built on top of an extensible Java library, functions as a computer algebra system that reduces user-specified arbitrary structural expressions
over pairs of SN nodes to comprehensible normal forms.
In a recent version of the tool \cite{FORTE25}, this calculus has been extended to operate at the matrix level (i.e. over the entire SN), thus facilitating the computation of more general structural dependencies, including the handling of priorities.

In this paper, we examine the effectiveness of the algebraic calculus underlying the \SNex\ tool for symbolic verification of structural \emph{invariants} of an \emph{extension} of SN (denoted ESN) that exhibits advantageous algebraic properties, with particular emphasis on (semi)flows. Specifically, we exploit the fact that the ESN arc functions constitute a class that is closed under fundamental functional operators, most notably composition \cite{pnse2021} and transposition. To the best of our knowledge, no alternative methodology has been proposed for the verification and computation of symbolic invariants without imposing constraints on SN color functions. In contrast, our approach is grounded in an extension of the SN formalism.
The main goal is to show that the symbolic computation engine implemented in \SNex\ efficiently supports semi-automatic formal verification of structural invariant properties, complementing the computation of structural dependencies.

In particular, we treat four distinct yet related topics.
We first address the verification of symbolic (semi-)flows, which currently cannot be directly generated or verified for general SN models, but only for certain restricted subclasses. We subsequently introduce a variant, referred to as colored flows, which abstracts away from quantitative characteristics (i.e., token multiplicities) and concentrates solely on qualitative properties (i.e., color types); in contrast to general symbolic flows, colored flows do not rely on conservative assumptions.
Furthermore, we discuss how, in principle, one could derive a generative family (basis) of semiflows by exploiting the constant-size property of ESN functions and by relating symbolic semiflows to conventional semiflows of the underlying skeleton (a P/T net). Finally, we outline a more general approach based on standard inductive inference techniques, whose full automation, however, necessitates extensions of the supporting tool. Although the exposition adopts a formal style, all concepts are elucidated by means of illustrative examples.

\paragraph{Related work}
The computation of generative families of flows—especially positive flows—for classical Petri nets (PN), i.e., integer solutions of a homogeneous linear system whose coefficient matrix is the incidence matrix of the net, has been extensively studied and solved effectively since \cite{colom1989}. Extending these techniques to High-Level Petri Nets (HLPN), in particular the automated verification and, when possible, computation of generative families of symbolic flows, has also been widely investigated. However, significant results exist only for restricted HLPN subclasses, mostly based on symbolic variants of classical elimination procedures (e.g., Gaussian or Gauss–Jordan elimination).
For example, \cite{Couvreur1991} gives an algorithm for generative families of flows in Commutative High Level Nets, a subclass of Colored Petri Nets (CPN) where color functions form a ring of commutative endomorphisms. This method was later extended to Unary Regular Nets and unary Predicate/Transition nets in \cite{CouvHad93}. The method proposed in \cite{Silva1991}, which relies on generalized inverses of (linearly extended) arc functions, is, in principle, applicable to general CPNs. However, it has never been automated due to the practical infeasibility of computing these generalized inverses in symbolic form. Since then, only a few substantial theoretical advances have appeared; among them, \cite{Eva2007} proposes a block-matrix–based algorithm for a subclass of Symmetric Nets (SN), namely Simple Well-formed Nets. Strict limitations on color functions (notably the lack of guards) make this technique unusable in many situations, including most examples in this paper.

\medskip

The balance of the paper: In section \ref{sec:SNdef}, we present a formal description of the ESN syntax, preceded by a brief introductory explanation. Section \ref{sec:eq-tr} establishes some fundamental properties of the ESN color functions language. In section \ref{sec:semiflow-def}, we introduce symbolic (semi)flows and present some basic properties. The central section \ref{sec:semiflow-ver} illustrates (semi)flow verification through a series of examples that demonstrate the key features of the ESN syntax. In section \ref{sec:basis}, we briefly discuss the calculation of a semiflow basis for ESN.
In section \ref{sec:cflow}, we also define a class of colored flows where the multiplicities of colors are ignored. Section \ref{sec:genInv} briefly addresses the verification of more general structural invariants for ESN. We conclude by summarizing our contributions and indicating directions for ongoing work.


%% file: SNdefinition.tex
\section{(Extended) Symmetric Nets}
\label{sec:SNdef}

Assuming only basic HLPN knowledge, we first give an informal overview of SN syntax using an example that will reappear later. We then formally present an extended version of SN, called ESN, focusing on the components needed for the invariant calculus. Stochastic parameters and priorities are omitted, as they are irrelevant in this context.

\begin{figure}[!h]
    \centering
    \includegraphics[width=\linewidth]{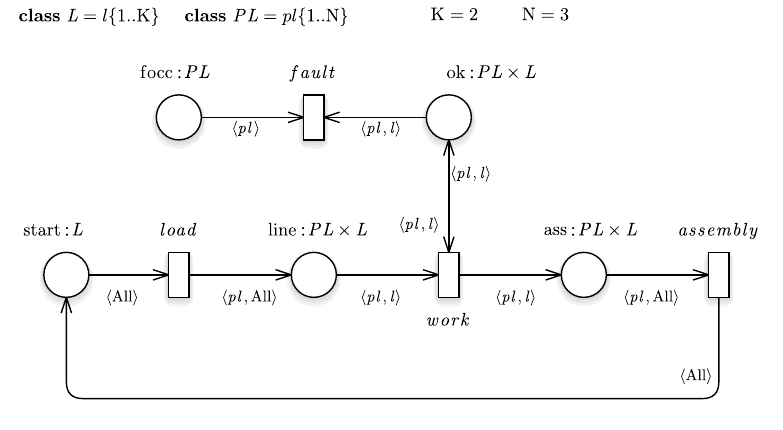}
    \caption{Distributed Production Line}
    \label{fig:PL}
\end{figure}

The domains of (E)SN nodes are Cartesian products of finite \emph{color classes}. These classes may be partitioned into static subclasses or, alternatively, ordered circularly. Arc inscriptions are sums of tuples of basic color functions, such as projections and constants denoting subclasses or entire classes.

Consider the SN depicted in Fig. \ref{fig:PL}, which models the core component of a distributed production system that, due to faults, experiences a controlled graceful degradation (we ignore adaptation aspects here). This SN is used as a reference in \cite{CAPRA-TCS2024,ICLP24}.
The system consists of $N$ parallel replicas of a production line (PL), each comprising $K$ similar interchangeable components that may fail. The PT net corresponding to the unfolding of this SN (for $K = 2, N = 2$) is reported in the appendix (Fig. \ref{fig:unfolding}).

The SN model thus uses two color classes, $L$ and $PL$, which are neither partitioned nor ordered, and whose cardinalities are $K$ and $N$, respectively (these are the parameters of the model).  
The places `\p{line}', `\p{ass}', and `\p{ok}' have domain $PL \times L$; place `\p{start}' has domain $L$, and place `\p{focc}' has domain $PL$. The places `\p{start}', `\p{line}', and `\p{ass}' represent the normal production cycle, whereas `\p{ok}' and `\p{focc}' capture the occurrence of a fault in a line.

The inferred domains for the transitions are $PL \times L$ for $work$ and $fault$, and $PL$ for $load$ and $assembly$. Consider, for example, the transition $work$: the function tuples on its incident arcs are given by $\tuple{pl,l}$, so a particular instance $\tuple{pl=c,l=c'}$ is enabled whenever a token $\tuple{c,c'}$ exists in both input places `\p{line}' and `\p{ok}'; In that case, when it fires, one such token is removed from `\p{line}' and inserted into `\p{ass}'.
The transition $load$ has domain $PL$: its input arc from `\p{start}' is labeled $\tuple{All}$, and its output arc to `\p{line}' is labeled $\tuple{pl,All}$; therefore, an instance $\tuple{pl=c}$ requires the presence of every color of $L$ in place `\p{start}'; firing this instance removes all these colors from `\p{start}' and produces in `\p{line}' the multiset $\sum_{c'\in L}\tuple{c,c'}$. Since the variable $pl$ appears free in the output arc, all instances enabled by $load$ are pairwise in conflict (see the corresponding transitions of the unfolding, shown in Fig. \ref{fig:unfolding}). An appropriate initial marking is given by $\vect{m}_0(\p{start}) = r \ \tuple{All}$, $\vect{m}_0(\p{focc}) = \tuple{All}$, and $\vect{m}_0(\p{ok}) = \tuple{All,All}$, with $r \in \Nat^+$. This implies, in approximate terms, that the production system is initialized with a number of row units that is a multiple of \(K\) (the number of lines in a production line, PL), and that all lines are operational. It can be readily demonstrated that, when faults occur in each PL, the system ultimately evolves into a deadlock state.

The other syntactical features, such as ordered/partitioned classes, guards, algebraic sums, and ESN-specific constructs, will be used in subsequent examples.

\subsection{Multiset functions and operations} 
Let $C$ be a non-empty set. A multiset in $C$ is a map $\vect{b}:C\rightarrow\Int$, 
where $\vect{b}(c)$ is the multiplicity of $c$;
its ``size'' is $|\vect{b}| := \sum_{c \in C} \vect{b}(c)$ (can be negative). $\Gbag{C}$ denotes (the set of) multisets and $\Bag{C}\subset\Gbag{C}$ those with natural multiplicities. If $\vect{b}\in\Bag{C}$ then its support is $\overline{\vect{b}}:=\{c\in C\mid \vect{b}(c) \neq 0\}$.
The empty multiset is denoted by $\vect{0}$.

Multiset operations are defined pointwise. Let $\vect{a}, \vect{b}\in\Gbag{C}$. For $\lambda\in\Int$, $\lambda\vect{a}(c)=\lambda *\vect{a}(c)$; binary operations $\mu$ 
move pointwise to $\vect{a}\,\mu\,\vect{b}$. 
For example $\vect{a} +\vect{b} (c) = \vect{a}(c) + \vect{b} (c) \ \forall c \in C  $; $\vect{a} -\vect{b} \equiv \vect{a} + (-1)\vect{b}$. The operators $\cap$ ($\equiv \ min$) and $\ominus$ are used in natural multisets, where $\vect{a}\ominus \vect{b}:=min(\vect{a}-\vect{b},\vect{0})$. The Cartesian product is also lifted: if $\vect{a}'\in\Gbag{C'}$, then $\tuple{\vect{a},\vect{a}'}\in\Gbag{C\times C'}$ and $\tuple{\vect{a},\vect{a}'}(c,c') := \vect{a}(c) * \vect{a}'(c')$. The relational operators are extended in an analogous way; for instance, $\vect{a} < \vect{b} := \bigwedge_{c \in C} \vect{a}(c) < \vect{b}(c)$.
Associative operations ($+(-), \cap, \tuple{\ldots}$) can be straightforwardly generalized to their n-ary counterparts.

Functional operators are obtained by evaluation. If $f,h\in\fset{A}{\Gbag{C}}$, $f'\in\fset{A}{\Gbag{C'}}$, and $\lambda\in\Int$, then $\lambda f$, $f+h$, $\tuple{f,f'}$, comparisons, $\cap$, $\ominus$, and support are defined pointwise.
For example, $f + h \in\fset{A}{\Gbag{C}} := f(a) + h(a)$, $\forall a \in A$. The product $\tuple{f,f'} \in \fset{A}{\Gbag{C\times C'}}$ (called a function tuple) is $\tuple{f,f'}(a) := \tuple{f(a),f'(a)}$, $\forall a$. Transposition is central: for $f\in\fset{A}{\Gbag{C}}$, $f^t\in\fset{C}{\Gbag{A}}$ is $f^t(c)(a) := f(a)(c)$, $\forall a, c$.

Functions are linearly extended to multiset arguments. If $f\in\fset{A}{\Gbag{B}}$, then, abusing the notation $f(\vect{a}):=\sum_{a\in A}\vect{a}(a)f(a)$; consequently, if $g \in\fset{B}{\Gbag{C}}$ then $g\circ f$ is defined by applying the linear extension of $g$ to $f(a)$. Guards between brackets are interpreted as functions $[p]\in\fset{A}{\Gbag{A}}$ that return $a$ ($1a$) when $p(a)$ is satisfied and $\vect{0}$ otherwise. For $f\in\fset{A}{\Gbag{B}}$, suffix and prefix guards are $f[p]:=f\circ[p]$ and $[p']f:=[p']\circ f$; prefix guards play a crucial role because they act as filters.

\subsection{Extended SN}
\label{subsec:ESN}
\newcommand{\pri}{\ensuremath \pi}
\begin{definition}[Extended SN] An ESN is a tuple:
$$ {\EuScript N} = (P,T,\ColClasses,\cdom, I,O,H,\Guards,\mk{m}_0)$$
where $P$ and $T$ are finite, nonempty, disjoint sets holding the places and the transitions; $\ColClasses$ is the set of color classes; ($\neq \emptyset)$; ${\cdom}$ assigns a color domain to each $v \in P \cup T$; $I, O$ and $H$ are families of arc functions defined for each $(p,t) \in P \times T$;
\Guards\ assigns each $t \in T$ a guard;
a marking $\mk{m}$ is a $P$-vector such that $\mk{m}(p) \in  Bag[\cdom(p)]$; $\mk{m}_0$ denotes the initial marking.
\end{definition}

\begin{itemize}
\item 
The color classes in $\ColClasses = \{C_i, i = 1 \, \ldots n\}$ are finite pairwise disjoint sets of elements called {\it colors}; each class $C_i$ can be {\it partitioned} into static subclasses $\{C_{i,j}\}_{j:1\ldots s_i}$ or alternately {\it ordered} circularly; if $C_i$ is ordered, then $!^m(c)$ ($m \in \Int$) denotes the successor/predecessor $m^{th}$  of $c \in C_i$ ($mod_{|C_i|}$).
\item 
 A color domain $\cdom_{\ColClasses}$ is a Cartesian product $\bigotimes_{i= 1}^n C_i^{e_i}$, where $e_i \in \Nat$ denotes the repetitions of class $C_i$ ($C_i^0$ is the identity); $|\cdom_{\ColClasses}| := \sum_{i= 1}^n e_i$. We define a partial order in the color domains: $\cdom_{\ColClasses} \leq \cdom'_{\ColClasses}$ if and only if $\forall i, \, e_i \leq e_i'$.
\item 
The color domain $\cdom(p)$ of a place $p$ defines the shape of tokens (color tuples) that can exist in $p$; the marking of $p$ is a multiset in ${\cdom}(p)$;
\item 
The color domain $\cdom(t)$ of a transition $t$ defines the instances of $t$; $\cdom(t)$ can be inferred from the inscriptions of the arcs surrounding $t$;
\item 
The guard $\Phi(t)$ of a transition $t$ is a function $\cdom(t) \rightarrow Bool$ syntactically expressed through a {\em standard predicate} (defined below); the color domain of $t$ is \emph{restricted} to $\{c \in  \cdom(t) | \Phi(t) = true\}$. For simplicity, the symbol \cdom(t) will denote such a restricted domain; $c \in \cdom(t)$ will also be denoted $(t,c)$; 

\item 
Let $W \in \{I,O, H\}$. An arc function $W(p,t) \in  \fset{\cdom(t)}{Bag[\cdom(p)]}$
belongs to the language $\lang_{\cdom(t),\cdom(p)}$ that is formally defined as:
\begin{equation}
\label{eq:arcfun}
 \lang_{\cdom_{\ColClasses},\cdom'_{\ColClasses}} := \{ F := \sum_i \lambda_i . [g_i'] T_i [g_i],\,\lambda_i \in \Int \}  \cup \{0_{\cdom_{\ColClasses},\cdom'_{\ColClasses}}\}
\end{equation}

$T_i$ is a \emph{function tuple} $\langle f_1, \ldots , f_{|\cdom'_{\ColClasses}|} \rangle$, and $g_i$ and $g_i'$ are \emph{standard predicates} defined in $\cdom_{\ColClasses}$ and $\cdom'_{\ColClasses}$, respectively.
Each $f_l$ in $T_i$, called \emph{a class function}, is a map  $\cdom_{\ColClasses}\rightarrow Bag^*[C_{h}]$, where $C_{h} \in \Sigma$ occurs at the position $l$ in $\cdom'_{\ColClasses}$. For a certain $\cdom_{\ColClasses}$ and $C_h \in \ColClasses$,
a class function is expressed as

\begin{equation}
\label{eq:classfun}
f \in \fset{\cdom_{\ColClasses}}{\Gbag{C_h}} := \sum_{k} \alpha_k. \epsilon_k, \,\, \alpha_k \in \Int
\end{equation}

\noindent $\epsilon_k \in \fset{\cdom_{\ColClasses}}{\Bag{C_h}}$ can be:
a projection $x_h^j$, $1 \leq j \leq e_h$ ($x_h^j(c)$ results in the $j^{th}$ color of class $C_h$ in the color tuple $c$), possibly preceded by $!^m$ if $C_h$ is ordered; a constant function $All_h$, which maps to $C_h$ (seen as a multiset); a constant function $All_{h,q}$, which maps to $C_{h,q}$, if $C_h$ is partitioned.

$g_i \in \fset{\cdom_{\ColClasses}}{Bool}$ is defined in terms of the basic clauses: $x_h^j = (\neq) x_h^y$, $x_h^j = (\neq) !^m x_h^y$ if $C_h$ is ordered and $x_h^j \in C_{h,q}$ if $C_h$ is partitioned.

Note that in (\ref{eq:arcfun}) and (\ref{eq:classfun}) there can be negative terms, but arc-functions must always map to $\Bag{\cdom(p})$.
\end{itemize}

\paragraph{Semantics} A transition instance $(t,c)$ is \textit{enabled} in the marking $\mk{m}$ if: 
$$\forall p \in P,  \Inp{p,t}(c) \leq \mk{m}(p) \wedge (\Inh{p,t}(c)  = 0 \vee \Inh{p,t}(c) > \mk{m}(p))$$
If $(t,c)$ is \textit{enabled} in $\mk{m}$, it can fire, leading to 
$\mk{m}' := \mk{m} + \Out{p,t}(c) - \Inp{p,t}(c)$ (denoted
$\mk{m} [(t,c) > \mk{m}'$).

\paragraph{Conventions} In the following, we use capital words, e.g., $C, L, PL$, to denote basic color classes in $\Sigma$, and the corresponding lowercase words to denote projections (variables) of the respective class (with a subscript if that class is repeated in a color domain): for example, if $\cdom_\ColClasses = L^2\times PL$, then the projection symbols in that domain are $\{l_1,l_2,pl\}$.

\vspace{-5pt}
\paragraph{ESN Arc function example} Consider the following function, which appears in Figure \ref{fig:broadcast}, whose semantics will be elucidated in Section \ref{sec:semiflow-ver}. This function cannot be represented within the original SN syntax.
 $$\tuple{All - n, n, ack} + [n_1 \neq n_2]\, \tuple{All - n,All -  n, ack} \in \fset{N}{N^2 \times L}$$
 This function occurs on the input arc from the place $\p{delivered}$, whose domain is $N^2 \times L$, to the transition $complete$, whose inferred domain is simply $N$. The color class $L$ is divided into $L_1 := \{data\}$ and $L_2 := \{ack\}$; here, the symbol $ack$ stands for $All_{L_2}$. The color class $N$ consists of $K$ colors (a model parameter) that are not distinguishable individually. For a given color $c \in N$, the function returns the (multi)set containing all triples $\tuple{c',c,ack}$ for every $c' \neq c$ (corresponding to the first term), together with all triples $\tuple{c'',c''',ack}$ for arbitrary $c''$, $c'''$ such that $c'' \neq c$, $c''' \neq c$, and $c''' \neq c''$ (the second term).
A central aspect is determining the size of the images of the arc functions, a nontrivial task due to possible filters, for which chromatic polynomials of graphs are used \cite{pnse2021}.

\section{Base Properties Of Arc Functions And Equivalence Transformations}
\label{sec:eq-tr}
The symbol $\lang$ shall henceforth refer to the language of arc functions (\ref{eq:arcfun}), whose (co-)domains are built over the set of color classes \ColClasses\ associated with a given ESN. 
In this section, we outline the main algebraic properties of $\lang$ and introduce basic structural transformations that support its symbolic manipulation.

\begin{proposition}
$\lang$ is closed under the base functional operations; composition $(\circ)$, transposition $({}^t)$, algebraic sum $(+(-))$, difference $(\ominus)$ and intersection $(\cap)$.
\end{proposition}
The same property holds for the derived language $\overline{\lang} := \{\overline{F} \mid F \in \lang\}$, consisting of functions in $\fset{\cdom_\ColClasses}{2^{\cdom'_\ColClasses}}$, where functional operators apply to sets. The base structural dependencies among the SN nodes can be expressed symbolically in $\overline{\lang}$, whose syntax matches that of $\lang$, except that class functions can use the operator $\cap$, the scalars are restricted to $\Nat$, and only one difference symbol ($\ominus$) is allowed.

It may be convenient to translate the SN annotations into "equivalent" forms that meet some requirements. In particular, (E)SN arc functions should be assumed to be of constant-size:
\begin{definition}
\label{def:c-size}
$F \in \fset{A}{\Gbag{B}}$ is constant-size if $\exists k \in \Int \, \forall a, |F(a)| = k$.   
\end{definition}

\begin{proposition}
An SN transition $t$ can be split into $t'_1,\ldots,t'_k$ that adopt uniquely constant size arc functions,
such that $\cdom(t) = \bigcup_{i:1}^k \cdom(t'_i)$, $\cdom(t'_i)$ are pair-wise disjoint, and $\mk{m} [(t,c) > \mk{m}'$ if and only if $\mk{m} [(t'_i,c) > \mk{m}'$ for some $t'_i$.   
\end{proposition}

\noindent An example is presented in Fig. \ref{fig:tequiv}: The function $\Out{\p{p}_1,t}$ has a size that depends on the guard of one term; therefore, $t$ is split (on the right) so that the functions satisfy the constant-size requirement.

\begin{figure}
    \centering
    \includegraphics[width=0.7\linewidth]{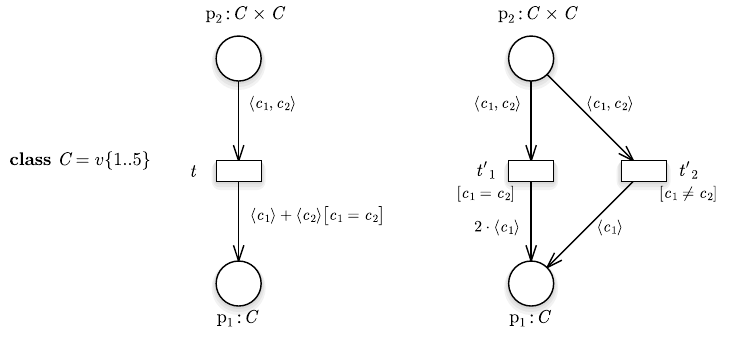}
    \caption{Transformation into constant-size functions}
    \label{fig:tequiv}
\end{figure}

We can also express any $W \in \lang$ as $\sum_i \lambda_i F_i$, where the terms $F_i$ are pair-wise disjoint and each $F_i$  produces constant-size multisets in which all element multiplicities are equal to one.
For example, after performing the transformation, we can readily confirm that the codomain of $W$ is $\Bag{B}$ by checking that all coefficients $\lambda_i$ belong to $\Nat$.

A concluding observation regards the equivalence test. In general, expressions $W$ in $\lang$ do not admit a canonical representative, although their constituent components $F_i$ do. However, for any $W, W' \in \lang_{\cdom_{\ColClasses},\cdom'_{\ColClasses}}$, the equivalence $W \equiv W'$ can be decided syntactically by examining $W - W'$ and $W' - W$. In fact, our symbolic calculus ensures that null expressions have a reduced syntactic form $0$.